\documentclass[conference]{worldcomp}

\usepackage[hmargin=.75in,vmargin=1in,left=0.7in,right=0.7in]{geometry}
\usepackage[american]{babel}
\usepackage[T1]{fontenc}
\usepackage{times}
\usepackage{caption}

\usepackage{textcomp}
\usepackage{epsfig,graphicx}
\usepackage{xcolor}
\usepackage{amsfonts,amsmath,amssymb}
\usepackage{fixltx2e} 
\usepackage{booktabs}

\usepackage{stmaryrd}
\usepackage{float}
\usepackage{algorithm}
\usepackage{algcompatible}
\usepackage{xspace}

\newtheorem{theorem}{Theorem}

\newtheorem{example}[theorem]{Example}
\newtheorem{lemma}[theorem]{Lemma}

\algnewcommand\AND{\textbf{and}\xspace}
\algnewcommand\OR{\textbf{or}\xspace}

\title{\bf Independent Spanning Trees in Eisenstein-Jacobi Networks}           

\author{
{\bfseries Z. Hussain$^1$, H. AboElFotoh$^1$, and B. AlBdaiwi$^1$}\\
$^1$Computer Science Department, Kuwait University, Kuwait\\
}

\begin{document}

\maketitle                        

\begin{abstract}
Spanning trees are widely used in networks for broadcasting, fault-tolerance, and securely delivering messages.
Hexagonal interconnection networks have a number of real life applications. Examples are cellular networks, computer graphics, and image processing. Eisenstein-Jacobi (EJ) networks are a generalization of hexagonal mesh topology. They have a wide range of potential applications, and thus they have received researchers’ attention in different areas among which interconnection networks and coding theory.
In this paper, we present two spanning trees’ constructions for Eisenstein-Jacobi (EJ). The first constructs three edge-disjoint node-independent spanning trees, while the second constructs six node-independent spanning trees but not edge disjoint. Based on the constructed trees, we develop routing algorithms that can securely deliver a message and tolerate a number of faults in point-to-point or in broadcast communications. The proposed work is also applied on higher dimensional EJ networks.
\end{abstract}

\vspace{1em}
\noindent\textbf{Keywords:}
 {\small  Interconnection network, hexagonal network, Eisenstein-Jacobi, spanning tree, edge disjoint, fault-tolerant, routing, broadcasting.} 

\section{Introduction}\label{sec:introduction}
The characteristics and properties of an interconnection network play a major role in the performance of the network since they determine the fault tolerance capabilities. Over past decades, many types of interconnection networks have been discussed such as Hypercube \cite{hayes1989hypercube}, mesh \cite{kumar1994introduction}, Torus \cite{dally1986torus}, $k$-ary $n$-cube \cite{bose1995lee}, butterfly, and Gaussian \cite{5204080}. Some machines have been implemented based on the topologies of these interconnection networks such as the IBM BlueGene \cite{adiga2002overview}, the Cray T3D and T3E \cite{scott1996cray}, the HP GS1280 multiprocessor \cite{cvetanovic2003performance}, and the J-machine \cite{noakes1993j}. Hexagonal networks are another type of interconnection are used in cellular networks \cite{nocetti2002addressing}, computer graphics \cite{lester1984computer}, image processing \cite{rummelt2011array}, and HARTS project \cite{shin1991harts}.

Eisenstein-Jacobi networks (EJ) were proposed in \cite{Martinez:2008:MHC:1375406.1375412} and \cite{5204080}. They are generated based on EJ integers \cite{huber1994codes}. EJ networks are symmetric 6$-$regular networks and they are generalizations of the hexagonal mesh topology presented in \cite{46277}\cite{90246}. One of the advantages of these type of networks is that they are used as a new method for constructing some classes of perfect codes that are used to solve the problem of finding perfect dominating set \cite{Martinez:2008:MHC:1375406.1375412}\cite{huber1994codes}. In addition, there are some studies on the applications of EJ netowkrs such as routing, broadcasting, and Hamiltonian cycles \cite{5204080}\cite{hussain2015edge}. The detailed definition of EJ network is discussed in Section \ref{sec:background}.

Independent spanning trees are widely used to broadcast messages and to obtain routing paths between nodes in a network. Moreover, they are used in networks to offer a reliable communication \cite{itai1988multi}\cite{krishnamoorthy1987fault}. For example, given a regular network of degree $d$, we can tolerate a number of faulty nodes by constructing $d$ independent spanning trees so that the network will still be connected even with the existence of $d-1$ faulty nodes. In addition, independent spanning trees are used to securely deliver a message to the destination node \cite{rescigno2001vertex}\cite{yang2011broadcasting}. For instance, a message can be sliced into $d$ parts where each part travels in distinct path until all parts reach the destination node. A clear definition of independent spanning trees is described in Section \ref{sec:background}.

The three main contributions of this paper are as follows. First, we introduce a construction of six node-independent spanning trees (IST) in EJ networks. Second, we present 
a construction of three edge-disjoint node-independent spanning trees (EDNIST) in EJ networks. Note that both constructions can be also
applied in hexagonal networks. Third, we develop routing algorithms based on the constructed trees that can be used in fault-tolerant point-to-point routing, fault-tolerant broadcasting,
or in secure message distributions. The designed algorithms are unified in the sense that they can be initiated from any node in an EJ network due to
the network topology symmetry and node transitivity.

Throughout this paper, the terms vertices and nodes are used interchangeably. Similarly for edges and links; and, graph and network. The rest of this paper is organized as follow. In Section \ref{sec:background} we review some terminologies from graph theory and we briefly describe the EJ networks. Section \ref{sec:relatedworks} discusses some previous works related to the domain of this paper. We introduce the node-independent spanning trees and edge-disjoint node-independent spanning trees in EJ networks in sections \ref{sec:EDNIST} and \ref{sec:IST}, respectively. In Section \ref{sec:routing}, we present the routing algorithm. The simulation results are described in Section \ref{sec:experimentalResults}. In Section \ref{sec:STinHigherEJ}, we apply the proposed construction methods on higher EJ networks. Finally, the paper is concluded in Section \ref{sec:conclusion}.

\section{Background\label{sec:background}}
Based on graph theory, some definitions and properties of graph are reviewed in this section.
In addition, we briefly describe the topological properties of EJ networks.

Given a graph $G(V,E)$ such that $v$ is the set of $|V|$ vertices and $E$ is the set of $|E|$ edges. An edge is a direct connection between two vertices denoted as $(u,v)$, such that $u,v \in V$. A sequence of connected edges are called path. That is, a path $P(s,d)$  of length $|P(s,d)| = n$ from vertex $s$ to vertex $d$ in $G$ is a sequence of connected edges $(s,x_1), (x_1,x_2), \dots, (x_n,d)$ where the intermediate vertices are distinct. Two paths $P_1(u,v)$ and $P_2(u,v)$ are said to be independent if their intermediate vertices are mutually disjoint. A tree $ST(V',E')$ that is a subgragh of $G(V,E)$ where $V' \subseteq V$ and $E' \subseteq E$ is called spanning tree when it contains all the vertices of G, i.e., $V' = V$. Two or more spanning trees $ST_j$, for $j = 1, 2, \dots, n$, rooted at vertex $r$ are called independent spanning trees if $\bigcap_{j=1}^{n} (P_{ST_j}(r,u) \setminus \{r,u\}) = \phi$ for $u \in V$, where $P_{ST_j}(r,u)$ is a path from $r$ to $u$ in the $j^{th}$ spanning tree. Further, the trees which their edge sets are pairwise disjoint are called edge-disjoint node-independent spanning trees. That is, for all trees $ST_j(V,E_j')$, for $j = 1, 2, \dots, n$, we have $E_p' \cap E_q' = \phi$ for all $p \neq q$ such that $1 \leq p \leq n$ and $1 \leq q \leq n$. In a graph $G$, the distance (denoted as $D(u,v)$) between two vertices $u$ and $v$ is the number of edges along the shortest path $P(u,v)$ (the path with minimum length over all possible paths between $u$ to $v$). The diameter $k$ of the graph is known as the shortest distance between two most farthest vertices in graph $G$.

Eisenstein-Jacobi networks \cite{5204080} are based on EJ integers \cite{huber1994codes}\cite{Martinez:2008:MHC:1375406.1375412}, which can be modeled on planar graphs as a graph $EJ_\alpha(V,E)$ generated by $\alpha = a+b\rho$ such that $0 \leq a \leq b$, where $V = \mathbb{Z}[\rho]_\alpha$ is the vertex set modulo $\alpha$, which represents the nodes in the network; and $E = \{(A,B) \in V \times V : (A-B) \equiv \pm 1, \pm \rho, \pm \rho^2 \ mod \ \alpha \}$ is the edge set, which represents the network links. The set of \textit{Eisenstein-Jacobi integers} $\mathbb{Z}[\rho]$ is defined as:
\begin{equation*}
\mathbb{Z}[\rho] = \{x+y\rho \ | \ x,y \in \mathbb{Z} \}
\end{equation*}
where $\rho  = (1 + i\sqrt 3)/2$, and $i = \sqrt{-1}$.
It is known that $\mathbb{Z}[\rho]$ is a Euclidean domain and the norm of EJ
integer $\alpha = a + b\rho$ is given by $N(\alpha) = {a^2} + {b^2} + ab$ \cite{5204080}, which is the total number of the distinct vertices in the network under the residue class modulo $\alpha$. It can be seen that $\rho^2 = \rho - 1$, $\rho^3 = -1$, $\rho^4 = -\rho$, $\rho^5 = 1-\rho$, and $\rho^6 = 1$.

The EJ networks are regular symmetric networks of degree six since each node in EJ network has six neighbors. The nodes in the network are addressed by $x + y\rho$. Two nodes in the network are adjacent if and only if there is an edge between them, i.e., the distance between them is 1.

The distance distribution in the network is based on the distance of the nodes from the center node, usually node 0. That is, it is the number of nodes at distance $t$ from node 0 where $t > 0$. EJ networks are called dense EJ networks when they contain a maximum number of nodes at distance $k$ where $k$ is the diameter of the network. Usually, their generator is $\alpha = a + b\rho$ such that $b = a + 1$. Thus, the number of nodes at distance $t$ is 1 or $6t$, respectively, for $t = 0$ or $t > 0$. It can be concluded that the diameter of dense EJ networks is $k = a$ and the number of nodes $d(t)$ at distance $t$ is:
\begin{equation*}
d(t) = \left\{
{\begin{array}{*{20}{l}}
1&{if \ t = 0}\\
{6t}&{if \ 1 \le t \leq k}\\
\end{array}} \right.
\end{equation*}

\begin{example}
\label{ex_EJ34}
Fig. \ref{EJ34} illustrates the node distribution (white nodes) of EJ network generated by $\alpha = 3+4\rho$ where the center node is 0.
\end{example}

There are two types of links in the EJ networks. The links that reside within the network are called regular links, which connect two neighboring nodes either two of them are none boundary nodes or one of them is a boundary node and the other one is a none boundary node in the network. Whereas, the links that are not residing within the network are called wraparound links, which connect two neighboring nodes where both of them are boundary nodes in the network. Fig. \ref{ex_EJ34} illustrates these types of links where the regular links are represented by solid lines and the wraparound links are represented by dotted lines.

The wraparound links can be recognized either by tiling or by modulo operation. By tiling, we mean that placing the EJ network at the origin of a grid and consider it as a basic EJ network with its center node is 0; and then making tiles by copying the basic EJ network and placing its copies around it. By modulo operation, we use $mod$ operator after adding $\pm 1$, $\pm \rho$, or $\pm \rho^2$ to the EJ integers to get the corresponding nodes in the basic EJ network. Note that, we have removed the straight dotted lines from node 3 to describe them as wrapped edges in the following example. Also, we have kept the nodes of the tiles that are connected to the basic EJ network through the wraparound edges and the rest of tile nodes are removed. The nodes in different tiles of the network are represented in different gray colors.

\begin{example}
\label{ex_EJ34_wraparound}
Consider the node $3\rho$ in Fig. \ref{EJ34}.  The node $3\rho$ is connected to node $1+3\rho$, which its corresponding node is $-2-\rho$ in the basic EJ network, through $+1$ edge. That is, the resultant of adding $+1$ to node $3\rho$ and then taking the $mod \ \alpha$ is node $-2-\rho$. Similarly, the $+\rho$ and $+\rho^2$ edges connect the node $3\rho$ to nodes, in respective order, $4\rho$ and $3\rho+\rho^2$, which their corresponding nodes in the basic EJ network are $-3$ and $-3\rho^2$, respectively.
\end{example}

\begin{figure}[!ht]
\centering
\includegraphics[scale=1]{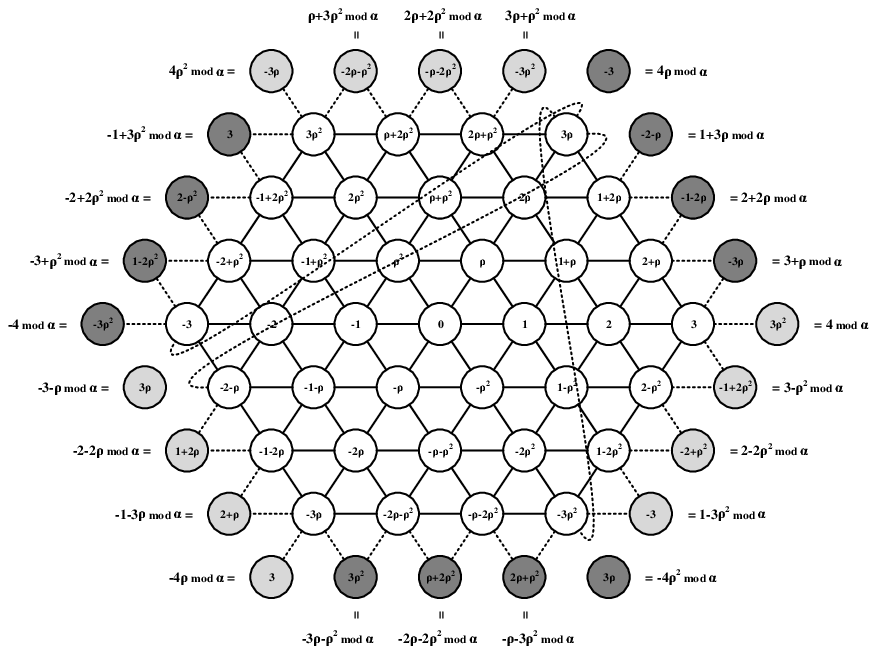}
\caption{EJ Network generated by $\alpha = 3+4\rho$ with dotted lines as wraparound edges.}
\label{EJ34}
\end{figure}

\section{Related Works\label{sec:relatedworks}}
Over the past years, the independent spanning trees have been widely studied in different types of networks. For instance, the construction of two completely independent spanning trees in any torus network and in the Cartesian product of any 2-connected graphs is investigated in \cite{hasunuma2012completely}. More studies on torus networks can be found in \cite{tang2010independent}\cite{tang2007parallel} and on Cartesian product graphs in \cite{ku2003constructing}\cite{yang2014optimal}. Additionally, The optimal independent spanning trees on Hypercubes is presented in \cite{tang2004optimal}. Further, a fully parallelized construction of ISTs on Mobius cubes has been discussed in \cite{yang2015fully}. Moreover, An implementation of a fast parallel algorithm for constructing ISTs on Parity Cubes is explained in \cite{chang2015fast}. In addition, in \cite{cheng2015dimensional}, the authors presented a common method for constructing ISTs on bijective connection networks based on V-dimensional-permutation technique. Furthermore, Building independent spanning trees on Twisted Cubes has been studied in
\cite{wang2012independent}\cite{Yang2014ASP}. There are some research studies on building ISTs in other networks such as: Crossed Cubes \cite{cheng2017constructing}, Locally Twisted Cubes \cite{Lin2010414}, Folded Hypercubes \cite{Yang20111254}\cite{Yang:2009:IST:1726593.1728973}, and Enhanced Hypercubes \cite{yang2015parallel}.

Our previous studies on independent spanning trees include the followings. In \cite{AlBdaiwi2016}, the two edge-disjoint node-independent spanning trees have been constructed for dense Gaussian networks. Further, in \cite{4ISTConference}\cite{hussain2017node}, the construction and parallel construction of four independent spanning spanning trees were presented such that the edges are not disjoint where the simulations have been done on the presence of 0, 1, 2, and 3 faulty nodes. Both studies have tree height $2k$, where $k$ is the diameter of the network. Lately, a parallel construction algorithms and its evaluations for edge-disjoint node-independent spanning trees in dense Gaussian networks was introduced in \cite{hussainparallel}.

\section{Edge-Disjoint Node-Independent Spanning Trees\label{sec:EDNIST}}

\subsection{Network Partitions}
\label{EDNISTNetworkPartitions}
Given EJ network generate by $\alpha = a+b\rho$ where $b = a+1$, the network can be partitioned into subsets as illustrated in Fig. \ref{EDNISTPartitions}. Let $c = 0, 2, 4$ for tree $t = 1, 2, 3$, respectively, and for $d = 1, 2, 3, 4, 5, 6$ such that $|x| + |y| = k$ where $k$ is the diameter of the network. Then, the subsets are as follows (all the powers of $\rho$ are modulo 6):

\noindent $B_d = \{ x\rho^{j-1} + y\rho^j \mid x > 0, y = 0, j = d+c \}$.

\noindent $T_d = \{ x\rho^{j-1} + y\rho^j \mid x > 0, y > 0, j = d+c \}$.

\noindent $S_2 = \{ x\rho^{j-1} + y\rho^j \mid x = a, y = 0, j = 2+c \}$.

\noindent $S_4 = \{ x\rho^{j-1} + y\rho^j \mid x = a, y = 0, j = 4+c \}$.

\noindent $L_4 = \{ x\rho^{j-1} + y\rho^j \mid x > 0, y = 1, j = 4+c \}$.

\noindent $L_6 = \{ x\rho^{j-1} + y\rho^j \mid x > 0, y = 1, j = 6+c \}$.

\noindent $B_2 \backslash S_2 = \{ x\rho^{j-1} + y\rho^j \mid 0 < x < k, y = 0, j = 2+c \}$.

\noindent $B_4 \backslash S_4 = \{ x\rho^{j-1} + y\rho^j \mid 0 < x < k, y = 0, j = 4+c \}$.

\noindent $T_4 \backslash L_4 = \{ x\rho^{j-1} + y\rho^j \mid x > 0, y > 1, j = 4+c \}$.

\noindent $T_6 \backslash L_6 = \{ x\rho^{j-1} + y\rho^j \mid x > 0, y > 1, j = 6+c \}$.

\hfill

%
\begin{lemma}
\label{EDNIST_subsets_are_disjoint}
The partitions in Fig. \ref{EDNISTPartitions} are disjoint and can be obtained from the above subsets.
\end{lemma}
\begin{proof}
Let $S$ be the set of subsets defined above and illustrated in Fig. \ref{EDNISTPartitions}, i.e. $S = \{ B_1$, $T_1$, $(B_2 \backslash S_2)$, $S_2$ $,$ $T_2$, $B_3$, $T_3$, $(B_4 \backslash S_4)$, $S_4$, $L_4$, $(T_4 \backslash L_4)$, $B_5$, $T_5$, $B_6$, $L_6$, $(T_6 \backslash L_6) \}$. Based on the definition of the subsets, 
for any two subsets $X,Y \in S, X \neq Y, X \cap Y = \phi$. 
\end{proof}

\begin{lemma}
\label{EDNIST_subsets_allnodes}
The subsets contains all nodes in the network.
\end{lemma}
\begin{proof}
Given the norm as a total number of nodes in the network, $N(\alpha) = a^2+b^2+ab$, then for $\alpha = k+(k+1)\rho$ we get $N(\alpha) = 3k^2 + 2k + 1$. It is obvious that $|B_d| = k$ for $d = 1,3,5,6$. Thus, we got a total of $4k$. In addition, $|S_2| = |S_4| = 1$, $|B_2 \backslash S_2| = |B_4 \backslash S_4| = k-1$, $|L_4| = |L_6| = k-1$.  Further, $|T_d| = \sum_{i=1}^{k-1}{\sum_{j=1}^{k-i} {1}} = \sum_{i=1}^{k-1}{(k-i)} = 1/2(k-1)k$ for $d = 1, 2, 3, 5$. That is, a total of $2(k-1)k$. Finally, we have $|T_4 \backslash L_4| = |T_6 \backslash L_6| = 1/2(k-1)k - (k-1)$. Thus, $B_d \cup T_d \cup S_2 \cup S_4 \cup L_4 \cup L_6 \cup (B_2 \backslash S_2) \cup (B_4 \backslash S_4) \cup (T_4 \backslash L_4) \cup (T_6 \backslash L_6) \cup \{0\}$ (including node 0) is equal to the set $V$, which is the set of nodes in the network. We conclude that, $4|B_d| + 4|T_d| + |S_2| + |S_4| + |L_4| + |L_6| + |B_2 \backslash S_2| + |B_4 \backslash S_4| + |T_4 \backslash L_4| + |T_6 \backslash L_6| + |\{0\}| = 3k^2 + 3k + 1 = N(\alpha)$ (excluding $B_2$, $B_4$, $T_4$, and $T_6$).
\end{proof}

\begin{figure}[h]
\centering
\includegraphics[scale=0.75]{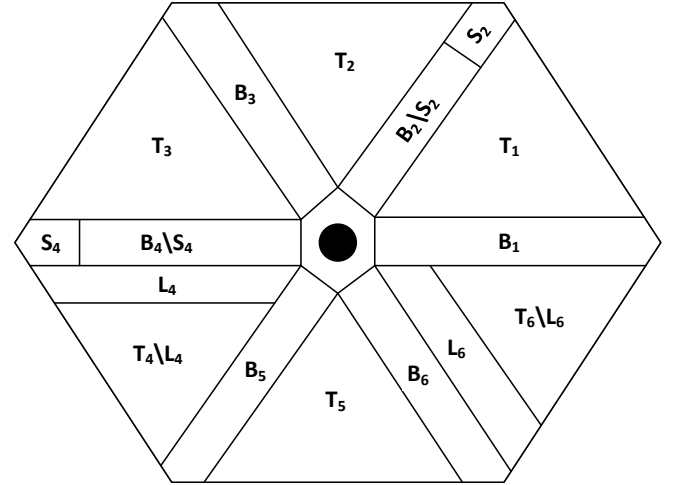}
\caption{EDNIST partitions.}
\label{EDNISTPartitions}
\end{figure}

This partitioning is helpful in finding the Edge-Disjoint Node-Independent Spanning Trees described in the following section.

\subsection{Tree Construction}
We construct the spanning tree based on Table \ref{ParentChildEDNIST}, which illustrates the parent and child nodes in the spanning tree for a given node belonging to a set.

\begin{example}
Given EJ network generated by $\alpha = 4+5\rho$ and a node $v=1+\rho$. For the first spanning tree, since $v \in T_1$, then its parent is node 1 and its child is node $1+2\rho$.
\end{example}

\begin{lemma}
\label{EDNIST_number_of_trees}
Let $ST_{ED}$ be a set of edge disjoint node independent spanning trees in $EJ$ network generated by $\alpha = a+b\rho$, where $b = a+1$, then $|ST_{ED}| \leq 3$.
\end{lemma}
\begin{proof}
The total number of nodes in the EJ network generated by $\alpha = a+b\rho$ is known as $N(\alpha) = a^2 + b^2 + ab$. In case of $b = a+1$, the total number of nodes is $3a^2 + 3a+1$ and
the total number of undirected edges is $9a^2+9a+3$. Since the spanning trees are edge disjoint then each spanning tree $ST_{ED}$ must have exactly $3a^2+3a$ undirected edges. Thus, it follows that $|ST_{ED}| \leq 3$.
\end{proof}

\begin{table}[h]
\centering
\caption{Parent and child nodes for EDNIST}
\label{ParentChildEDNIST}
\begin{tabular}{|c|c|c|}
\hline
Set                             & Parent       & Child                            \\ \hline
$B_1$                           & $\rho^{j+2}$ & $\rho^{j-1}, \rho^j, \rho^{j+4}$ \\ \hline
$B_2 \backslash S_2$            & $\rho^{j+3}$ & $\rho^j, \rho^{j+2}$             \\ \hline
$S_2$                           & $\rho^{j+3}$ & $\rho^{j+2}$                     \\ \hline
$B_3 \cup B_5$                  & $\rho^{j-1}$ & --                               \\ \hline
$B_4 \backslash S_4$            & $\rho^{j+1}$ & --                               \\ \hline
$S_4$                           & $\rho^{j+1}$ & $\rho^{j+2}$                     \\ \hline
$B_6 \cup T_2 \cup T_5$         & $\rho^{j-1}$ & $\rho^{j+2}$                     \\ \hline
$T_1 \cup (T_4 \backslash L_4)$ & $\rho^{j+3}$ & $\rho^j$                         \\ \hline
$T_3 \cup (T_6 \backslash L_6)$ & $\rho^{j+1}$ & $\rho^{j+4}$                     \\ \hline
$L_4$                           & $\rho^{j+3}$ & --                               \\ \hline
$L_6$                           & $\rho^{j+1}$ & $\rho^{j+2}, \rho^{j+4}$         \\ \hline
\end{tabular}
\end{table}

\begin{lemma}
\label{EDNISTisConnected}
The first spanning tree is connected.
\end{lemma}
\begin{proof}
Based on Section \ref{EDNISTNetworkPartitions}, consider the $j$ values with $c = 0$. Let $ST_{ED_1}(V_1,E_1)$ represents the first edge disjoint node independent spanning tree where $V_1 \subseteq V$ and $E_1 \subseteq E$ are the set of nodes and edges in $ST_{ED_1}$, respectively. Based on Lemma \ref{EDNIST_number_of_trees}, we have $|E_1| = 3a^2 + 3a = |V_1| - 1$. Further, Table \ref{EDNISTPaths} shows the path
from the source node $S = 0$ to all other nodes in the network using tree $ST_{ED_1}$. As it is noted in Table \ref{EDNISTPaths}, the paths are described by a word on the alphabet $\{-1, 1, -\rho, \rho, -\rho^2, \rho^2\}$ where the symbols denote the direction of the edges to be passed. The number of steps are represented as $(direction)^{steps}$. We conclude that $ST_{ED_1}$ is connected.
\end{proof}

\begin{example}
In the first spanning tree, let $S = 0$ and $D = \rho^4 + 3\rho^5$ (which is $D = -1-4\rho^2$) where $x = -1$ and $y = -4$, then $D \in B_6 \cup T_5 \cup B_5$. Thus, the steps are $1 (-\rho^2)^4 (-1)^2$. That is, $D$ can be reached by going $1$ step along direction 1, then $4$ steps along  direction $-\rho^2$, and finally $2$ steps along direction $-1$. 
\end{example}

\begin{lemma}
\label{EDNISTrotating}
The second and third spanning trees can be obtained by rotating the first spanning tree.
\end{lemma}
\begin{proof}
Based on Lemmas \ref{EDNIST_subsets_are_disjoint} and \ref{EDNIST_subsets_allnodes}, and Table \ref{ParentChildEDNIST}, since the network is symmetric then it is sufficient to prove that the obtained second and third spanning trees are connected by following Lemma \ref{EDNISTisConnected}, but with different $j$ values with $c = 2, 4$ as described in Section \ref{EDNISTNetworkPartitions}.
\end{proof}

\begin{theorem}
\label{EDNIST_tree_connected}
$ST_{ED_t}$, for $t = 1, 2, 3$, are edge disjoint node independent spanning trees.
\end{theorem}
\begin{proof}
Based on Lemmas \ref{EDNIST_subsets_are_disjoint}-\ref{EDNISTrotating}, and Tables \ref{ParentChildEDNIST} and \ref{EDNISTPaths},
let $ST_{ED_t}(E)$ be the set of undirected edges for spanning tree $t$.
Thus, we get $ST_{ED_t}(E) \cap ST_{ED_{t'}}(E) = \phi, t, t' \in \{1, 2, 3\}, t \neq t'$. We conclude that all trees are edge disjoint node independent spanning trees.
\end{proof}

\begin{table}[h]
\centering
\caption{Steps from node $S=0$ to all other nodes $D=x\rho^{j-1}+y\rho^j$, where $k$ is the diameter}
\label{EDNISTPaths}
\begin{tabular}{|c|c|c|}
\hline
Node in set                                            & Path (steps)   \\ \hline
$B_1$                                                    & $(1)^x$         \\ \hline
$T_1$                                                    & $(1)^x (\rho)^y$         \\ \hline
$\{B_2 \backslash S_2\} \cup S_2$        & $(\rho)^y$         \\ \hline
$T_2 \cup B_3$                                     & $(\rho)^{|y|} (-1)^{|x|}$  \\
                                                              & (after converting to form $x+y\rho$)         \\ \hline
$T_3 \{B_4 \backslash S_4\} \cup S_4$  & $(1)^{k-|x|+1} (-\rho^2)^{k-y}$         \\ \hline
$L_4 \cup \{T_4 \backslash L_4\}$        & $(1)^{k-|x|} (\rho)^{k-|y|+1}$         \\ \hline
$L_6 \cup \{T_6 \backslash L_6\}$        & $(1)^x (-\rho^2)^{y}$         \\ \hline
$B_6 \cup T_5 \cup B_5$                       & $(1) (-\rho^2)^{|y|} (-1)^{|x|+1}$ \\
                                                              & (after converting to form $x+y\rho^2$)        \\ \hline
\end{tabular}
\end{table}

\begin{lemma}
\label{EDNIST_tree_depth}
The depth of all trees $ST_{ED}$, for $t = 1, 2, 3$, is $2k+2$.
\end{lemma}
\begin{proof}
The proof is provided for tree $ST_{ED_1}$. The same proof can be applied to the other trees accordingly.
Based on Lemma \ref{EDNIST_subsets_allnodes} and Table \ref{EDNISTPaths}, the longest path in tree $ST_{ED_1}$ starting from node 0 is $2k+2$, which leads to node $-k\rho$ or to node $k\rho^2$.
Further, the last set in Table \ref{EDNISTPaths} is $B_6 \cup T_5 \cup B_5$ which has a maximum steps of $max |x| + max |y| + 1 + 1 = k + k + 2 = 2k +2$.
\end{proof}

\begin{figure*}[h]
\centering
\includegraphics[scale=0.60]{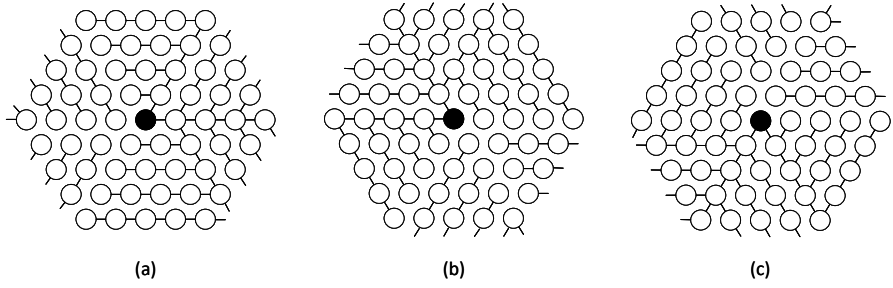}
\caption{First-third EDNISTs from (a) to (c), respectively, for EJ with $\alpha = 4+5\rho$.}
\label{3EDNISTs}
\end{figure*}

\begin{example}
Fig. \ref{3EDNISTs}(a), Fig. \ref{3EDNISTs}(b), and Fig. \ref{3EDNISTs}(c) illustrate the first, second, and third edge disjoint node independent spanning trees in EJ network generated by $\alpha = 4+5\rho$, respectively.
\end{example}

\section{Node-Independent Spanning Trees\label{sec:IST}}
This section discusses the construction of six node-independent spanning trees in EJ networks. First, we describe the network partitions in Section \ref{ISTNetworkPartitions}, which help in constructing these trees as illustrated in Section \ref{ISTTreeConstruction}.

\subsection{Network Partitions\label{ISTNetworkPartitions}}
The EJ network generated by $\alpha = a+b\rho$, where $b = a+1$ can be partitioned into disjoint subsets, as shown in Fig. \ref{ISTPartitions}. The disjoint subsets are described as follows. Let $c = t-1$ for tree $t = 1, 2, 3, 4, 5, 6$, $d = 1, 2, 3, 4, 5, 6$, and all the powers of $\rho$ are modulo 6. In addition, Let $|x| + |y| = k$ where $k$ is the network diameter, then:

\noindent $B_d = \{ x\rho^{j-1} + y\rho^j \mid x > 0, y = 0, j = d+c \}$.

\noindent $T_d = \{ x\rho^{j-1} + y\rho^j \mid x > 0, y > 0, j = d+c \}$.

\noindent $S = \{ x\rho^{j-1} + y\rho^j \mid x = 1, y = 0, j = 5+c \}$.

\noindent $B_5 \backslash S = \{ x\rho^{j-1} + y\rho^j \mid x > 1, y = 0, j = 5+c \}$.

\noindent $L_3 = \{ x\rho^{j-1} + y\rho^j \mid x > 0, y = 1, j = 3+c \}$.

\noindent $L_4 = \{ x\rho^{j-1} + y\rho^j \mid x > 0, y = 1, j = 4+c \}$.

\noindent $T_3 \backslash L_3 = \{ x\rho^{j-1} + y\rho^j \mid x > 0, y > 1, j = 3+c \}$.

\noindent $T_4 \backslash L_4 = \{ x\rho^{j-1} + y\rho^j \mid x > 0, y > 1, j = 4+c \}$.

\hfill

\begin{lemma}
\label{IST_subsets_are_disjoint}
The partitions in Fig. \ref{ISTPartitions} are disjoint and can be obtained from the above subsets.
\end{lemma}
\begin{proof}
Let $S$ be the set of subsets defined above and illustrated in Fig. \ref{ISTPartitions}, i.e. $S = \{ B_1$, $T_1$, $B_2$, $T_2$, $B_3$, $L_3$, $(T_3 \backslash L_3)$, $B_4$, $L_4$, $(T_4 \backslash L_4)$, $(B_5 \backslash S)$, $S$, $T_5$, $B_6$, $T_6 \}$. Based on the definition of the subsets, 
for any two subsets $X,Y \in S, X \neq Y, X \cap Y = \phi$. 
\end{proof}

\begin{lemma}
\label{IST_subsets_allnodes}
The subsets contains all nodes in the network.
\end{lemma}
\begin{proof}
Given the norm as a total number of nodes in the network, $N(\alpha) = a^2+b^2+ab$, then for $\alpha = k+(k+1)\rho$ we get $N(\alpha) = 3k^2 + 2k + 1$. It is obvious that $|B_d| = k$ for $d = 1,2,3,4,6$. Thus, we got a total of $5k$. In addition, $|S| = 1$, $|B_5 \backslash S| = k-1$, $|L_3| = k-1$, $|L_4| = k-1$.  Further, $|T_d| = \sum_{i=1}^{k-1}{\sum_{j=1}^{k-i} {1}} = \sum_{i=1}^{k-1}{(k-i)} = 1/2(k-1)k$ for $d = 1, 2, 5, 6$. That is, a total of $2(k-1)k$. Finally, we have $|T_3 \backslash L_3| = |T_4 \backslash L_4| = 1/2(k-1)k - (k-1)$. Thus, $B_d \cup T_d \cup S \cup (B_5 \backslash S) \cup L_3 \cup L_4 \cup (T_3 \backslash L_3) \cup (T_4 \backslash L_4) \cup \{0\}$ (including node 0) is equal to the set $V$, which is the set of nodes in the network. We conclude that, $5|B_d| + 4|T_d| + |S| + |(B_5 \backslash S)| + |L_3| + |L_4| + |(T_3 \backslash L_3)| + |(T_4 \backslash L_4)| + |\{0\}| = 3k^2 + 3k + 1 = N(\alpha)$ (excluding $|B_5|$, $|T_3|$, and $|T_4|$).
\end{proof}

\begin{figure}[h]
\centering
\includegraphics[scale=1]{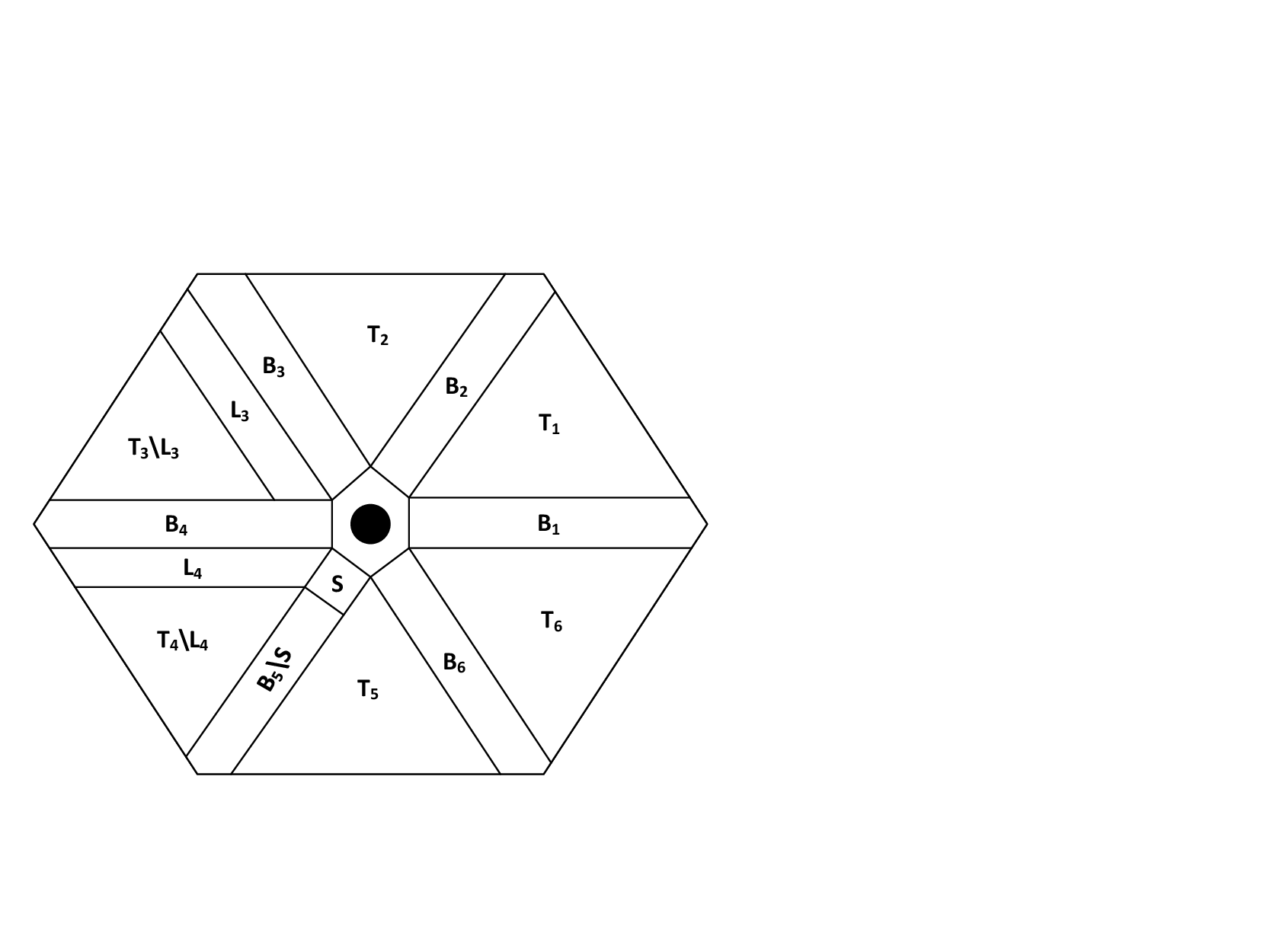}
\caption{IST partitions.}
\label{ISTPartitions}
\end{figure}

This partitioning is helpful in finding the Node-Independent Spanning Trees described in the following section.

\subsection{Tree Construction\label{ISTTreeConstruction}}
Similar to Section \ref{sec:EDNIST}, the node independent spanning trees can be constructed based on Table \ref{ParentChildIST}, which provides the parent and child nodes for a given node in a certain set.

\begin{example}
Given EJ network generated by $\alpha = 4+5\rho$ and a node $v=4\rho$. For the first spanning tree, since $v \in B_2$, then its parent is node $4\rho-1$ and it has no child.
\end{example}

\begin{lemma}
\label{IST_number_of_trees}
Let $ST$ be a set of node independent spanning trees in $EJ$ network generated by $\alpha = a+b\rho$, where $b = a+1$, then $|ST| \leq 6$.
\end{lemma}
\begin{proof}
Following Lemma \ref{EDNIST_number_of_trees}, and since the edges are not disjoint then the directed edges are used to construct the trees instead of undirected edges.
That is, using each undirected edge twice (in both directions) to construct two different trees that are not necessarily edge disjoint we get $2 |ST_{ED}| = |ST| \leq 6$.
\end{proof}

\begin{table}[h]
\centering
\caption{Parent and child nodes for IST}
\label{ParentChildIST}
\begin{tabular}{|c|c|c|}
\hline
Set                             & Parent       & Child                            \\ \hline
$B_1$                           & $\rho^{t+2}$ & $\rho^{t-1}, \rho^t, \rho^{t+4}$ \\ \hline
$B_2 \cup B_6$                  & $\rho^{t+2}$ & --                               \\ \hline
$B_3 \cup T_2 \cup T_5$         & $\rho^{t+2}$ & $\rho^{t-1}$                     \\ \hline
$L_3$                           & $\rho^{t+1}$ & $\rho^{t-1}, \rho^{t+4}$         \\ \hline
$(T_3 \backslash L_3) \cup T_6$ & $\rho^{t+1}$ & $\rho^{t+4}$                     \\ \hline
$B_4$                           & $\rho^{t+1}$ & --                               \\ \hline
$L_4$                           & $\rho^{t+3}$ & --                               \\ \hline
$T_1 \cup (T_4 \backslash L_4)$ & $\rho^{t+3}$ & $\rho^t$                         \\ \hline
$(B_5 \backslash S)$            & $\rho^{t+3}$ & $\rho^{t-1}, \rho^t$             \\ \hline
$S$                             & $\rho^{t+3}$ & $\rho^{t-1}$                     \\ \hline
\end{tabular}
\end{table}

\begin{lemma}
\label{ISTisConnected}
The first node independent spanning tree is connected.
\end{lemma}
\begin{proof}
Based on Section \ref{ISTNetworkPartitions}, consider the $j$ values with $c = t-1$, for $t=2,3,4,5,6$. Let $ST_1(V_1,E_1)$ represents the first node independent spanning tree where $V_1 \subseteq V$ and $E_1 \subseteq E$ are the set of nodes and edges in $ST_1$, respectively.
Based on Lemma \ref{IST_number_of_trees}, we get $|E_1| = 3a^2 + 3a = |V_1| - 1$. Further, Table \ref{ISTPaths} shows the path
from the source node $S = 0$ to all other nodes in the network using tree $ST_1$. As it is noted in Table \ref{ISTPaths}, the paths are described by a word on the alphabet $\{-1, 1, -\rho, \rho, -\rho^2, \rho^2\}$ where the symbols denote the direction of the edges to be passed. The number of steps are represented as $(direction)^{steps}$. We conclude that $ST_1$ is connected.
\end{proof}

\begin{example}
In the first spanning tree, let $S = 0$ and $D = \rho^4 + 3\rho^5$ (which is $D = 3-4\rho$) where $x = 3$ and $y = -4$, then $D \in \{B_5 \backslash S\} \cup S \cup T_5 \cup B_6$. Thus, the steps are $(1)^4 (\rho)^1 (1)^3$. That is, $D$ can be reached by going $4$ steps along direction 1, then $1$ step along  direction $\rho$, and finally $3$ steps along direction 1. 
\end{example}

\begin{lemma}
\label{ISTrotating}
The second, third, forth, fifth, and sixth node independent spanning trees can be obtained by rotating the first node independent spanning tree.
\end{lemma}
\begin{proof}
Based on Lemmas \ref{IST_subsets_are_disjoint} and \ref{IST_subsets_allnodes}, and Table \ref{ParentChildIST}, since the network is symmetric then it is sufficient to prove that the obtained second, third, forth, fifth, and sixth node independent spanning trees are connected by following Lemma \ref{ISTisConnected}, but with different $j$ values with $c = t-1$ for $t = 2,3,4,5,6$ as described in Section \ref{ISTNetworkPartitions}.
\end{proof}

\begin{theorem}
\label{IST_tree_connected}
$ST_t$, for $t = 1, 2, 3, 4, 5, 6$, are node independent spanning trees.
\end{theorem}
\begin{proof}
Based on Lemmas \ref{IST_subsets_are_disjoint}-\ref{ISTrotating}, and Tables \ref{ParentChildIST} and \ref{ISTPaths},
let $ST_t(E)$ be the set of directed edges for the spanning tree $t$.
Thus, we get $ST_t(E) \cap ST_{t'}(E) = \phi, t, t' \in \{1, 2, 3, 4, 5, 6\}, t \neq t'$. That is, each directed edge is used once among all trees.
We conclude that all trees are node independent spanning trees.
\end{proof}

\begin{table}[h]
\centering
\caption{Steps from node $S=0$ to all other nodes $D=x\rho^{j-1}+y\rho^j$, where $k$ is the diameter}
\label{ISTPaths}
\begin{tabular}{|c|c|c|}
\hline
Node in set                                                              & Path (steps)   \\ \hline
$B_1$                                                                      & $(1)^x$         \\ \hline
$T_1 \cup B_2$                                                       & $(1)^x (\rho)^y$         \\ \hline
$T_6$                                                                      & $(1)^x (-\rho^2)^y$         \\ \hline
$B_4 \cup \{T_3 \backslash L_3\} \cup L_3$           & $(1)^{k-|x|+1} (-\rho^2)^{k-y}$         \\ \hline
$B_3 \cup T_2$                                                       & $(1)^k (-\rho^2)^{k-y} (1)^{x+1}$         \\ \hline
$L_4 \cup \{T_4 \backslash L_4\}$                          & $(1)^{k-|x|} (\rho)^{k-y+1}$         \\ \hline
$\{B_5 \backslash S\} \cup S \cup T_5 \cup B_6$   & $(1)^k (\rho)^{k-|y|+1} (1)^x$         \\ \hline
\end{tabular}
\end{table}

\begin{lemma}
\label{IST_tree_depth}
The depth of all trees $ST_t$, for $t = 1, 2, 3, 4, 5, 6$, is $2k+1$.
\end{lemma}
\begin{proof}
The proof is provided for tree $ST_1$. The same proof can be applied to the other trees accordingly.
Based on Lemma \ref{IST_subsets_allnodes} and Table \ref{ISTPaths}, the longest path in tree $ST_1$ starting from node 0 is $2k+1$, which leads to nodes $c\rho$ or to nodes $-c\rho^2$, where $1 \leq c \leq k$.
\end{proof}

\begin{example}
Fig. \ref{6ISTs}(a), Fig. \ref{6ISTs}(b), Fig. \ref{6ISTs}(c), Fig. \ref{6ISTs}(d), Fig. \ref{6ISTs}(e), and Fig. \ref{6ISTs}(f) illustrate the first, second, third, fourth, fifth, and sixth node independent spanning trees in EJ network generated by $\alpha = 4+5\rho$, respectively.
\end{example}

\begin{figure*}[h]
\centering
\includegraphics[scale=0.20]{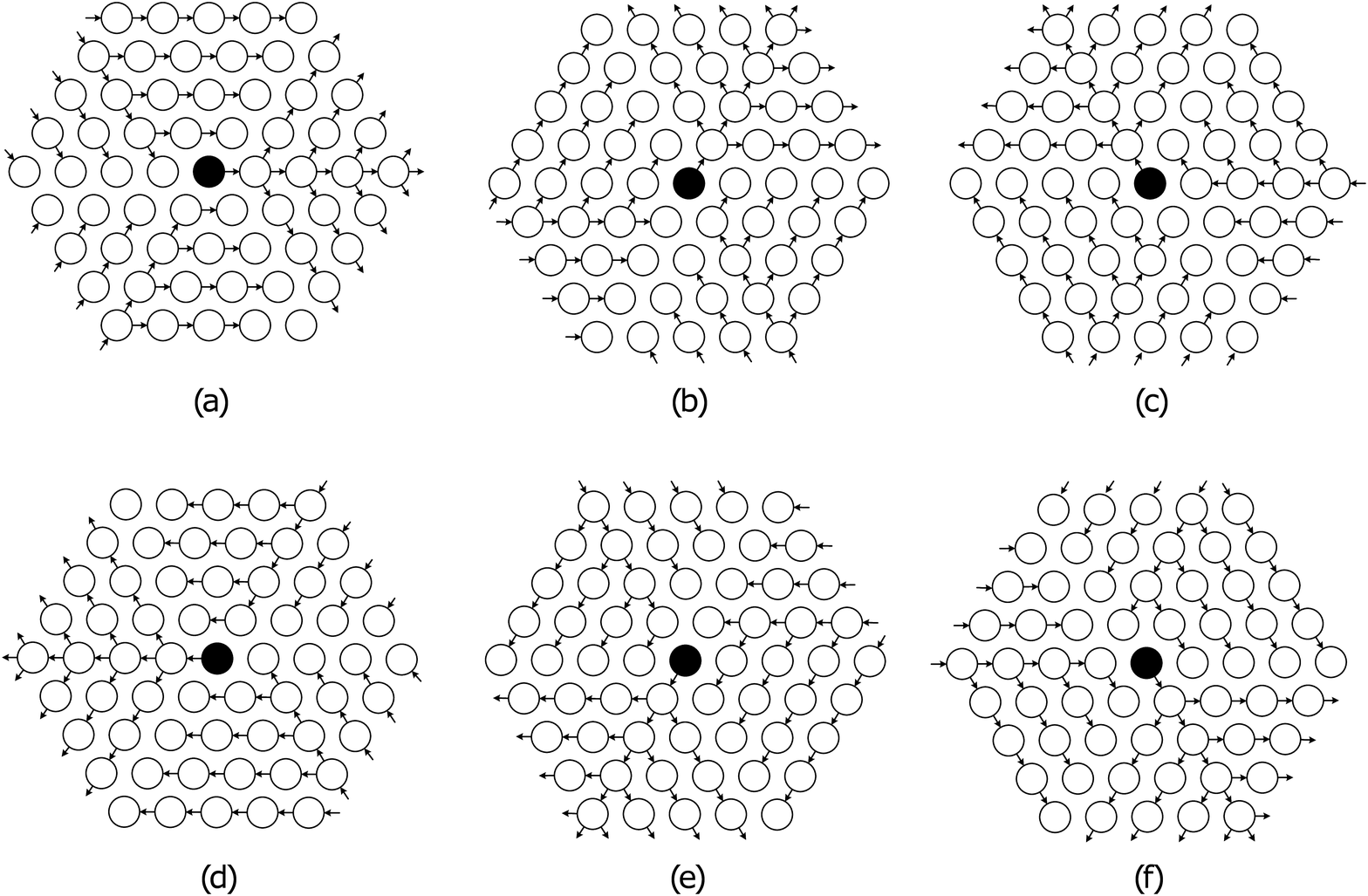}
\caption{First-sixth ISTs from (a) to (f), respectively, for EJ with $\alpha = 4+5\rho$.}
\label{6ISTs}
\end{figure*}

\section{Routing}
\label{sec:routing}
In this section, we present the algorithm used to rout the messages in the trees constructed in Sections \ref{sec:EDNIST} and \ref{sec:IST}.
The algorithm uses Tables \ref{EDNISTPaths} and \ref{ISTPaths} to determine the link in the current node to be used for sending/forwarding the messages.

Algorithm \ref{Init-Route} describes the procedures to be taken at the source node as follows. Since Tables \ref{EDNISTPaths} and \ref{ISTPaths} assume the source node is 0
and due to the symmetry of the network then, as stated in line 1, the given source node $S$ is mapped to node 0, and relatively, the destination node $D$ is also mapped. Line 2, obtains
the path sequence as tuples consisting of ($direction$, $steps$) based on the $Set$ that the destination node $D$ belongs to. The $direction$ represents the link to be used in the current
node to send/forward the message and the $steps$ is the number of hops along the given $direction$. In line 3, the first tuple is obtained to be used to send the message in line 4. The time complexity of this algorithm is $O(n)$, where $n$ is the total number of nodes in the network, since all the lines take constant time except the line 2, which needs to match the $D$ with its corresponding $Set$. The communication complexity is $O(1)$ since it only sends one message as stated in line 4.

\begin{algorithm}[H]
\caption{Init-Routing($S$, $D$, $\alpha$)}
\begin{algorithmic}[1]
\label{Init-Route}
\STATE Map $S$ to node 0 and $D$ accordingly
\STATE Lookup Table \ref{EDNISTPaths} or \ref{ISTPaths} such that $D \in Set$ to get the corresponding path $P$ consisting of a sequence of tuples $(direction, steps)$
\STATE $dir$:$steps$ = $P$.pop()
\STATE Send through link ($dir$) message Rout($dir$, $steps - 1$, $P$, $S + dir$, $D$, $\alpha$)
\end{algorithmic}
\end{algorithm}

In Algorithm \ref{Route}, Lines 1-4 checks whether the message has arrived to the destination node. Lines 5-7, checks whether the number of the steps is equal to 0. If
so, then it means that there are no more steps in current given direction. Thus, a tuple is obtained from the current path $P$ sequence where the remaining tuples will be
obtained later on. In line 8, the algorithm sends the message using link described in $direction$ and reduces the number of $steps$ by 1. The time complexity of this algorithm
is $O(1)$ since each line takes constant time. The communication complexity is $O(1)$ per node as stated in line 8.

\begin{algorithm}[H]
\caption{Routing($dir$, $steps$, $P$, $S$, $D$, $\alpha$)}
\begin{algorithmic}[1]
\label{Route}
\IF {($S = D \ mod \ \alpha $)}
\STATE Consume packet
\STATE Return
\ENDIF
\IF {($steps = 0$)}
\STATE $dir$:$steps$ = $P$.pop()
\ENDIF
\STATE Send through link ($dir$) message Rout($dir$, $steps - 1$, $P$, $S + dir$, $D$, $\alpha$)
\end{algorithmic}
\end{algorithm}


The following example illustrates the usage of the routing algorithm.

\begin{example}
\label{exampleRouting}
Let the source node be $S = 0$ and the destination node be $D = -2-\rho^2$ in EJ network generated by $\alpha = 4+5\rho$. We get $k = 4$, $x = -2$, and $y = -1$.
Based on Algorithm \ref{Init-Route}, no need to map $S$ because it is 0 and we obtain path $P = \{(1,3),(-\rho^2,3)\}$ since $D \in Set = T_3 \{B_4 \backslash S_4\} \cup S_4$.
The $dir$ is set to $1$ and the $steps$ is set to $3$ by calling $P.pop()$, which results $P = \{(-\rho^2,3)\}$. After that, based on Algorithm \ref{Init-Route}, the source
node $S$ sends the message $Route(1, 2, P, 1, -2-\rho^2, \alpha)$ through link $1$ to node $1$. Node $1$ applies the line 8 in Algorithm \ref{Route} and continue sending the
message $Route(1, 1, P, 2, -2-\rho^2, \alpha)$ to node $2$ via link $1$.
Node $2$ applies the line 8 in Algorithm \ref{Route} and continue sending the
message $Route(1, 0, P, 3, -2-\rho^2, \alpha)$ to node $3$ via link $1$.
At node $3$, since the $steps = 0 $ then it gets the next tuple by calling $P.pop()$ and sets $dir$ to $-\rho^2$ and $steps$ to $3$, after that it continue
sending the message $Route(-\rho^2, 2, P, 3-\rho^2, -2-\rho^2, \alpha)$ to node $3-\rho^2$ via link $-\rho^2$.
The receiving node $3-\rho^2$ applies the line 8 in Algorithm \ref{Route} and continue sending the
message $Route(-\rho^2, 1, P, 3-2\rho^2, -2-\rho^2, \alpha)$ to node $3-2\rho^2$ via link $-\rho^2$.
The receiving node $3-2\rho^2$ applies the line 8 in Algorithm \ref{Route} and continue sending the
message $Route(-\rho^2, 0, P, 3-3\rho^2, -2-\rho^2, \alpha)$ to node $3-3\rho^2$ via link $-\rho^2$.
Finally, the receiving node $3-3\rho^2$ observes that $S = 3-3\rho^2 \equiv -2-\rho^2 = D \ mod \ \alpha $ and receives the message.
\end{example}

\section{Experimental Results\label{sec:experimentalResults}}
In this section, we discuss the simulation results. We have used a Python network simulator called NetworkX \cite{hagberg2005networkx,hagberg2008exploring} in our
implementation. It is a package used to represent and analyze the networks and the algorithms used in the networks. In our simulation, we assumed that each node can send and receive messages simultaneously to all its neighbors.

Based on Section \ref{sec:IST}, the algorithm always constructs 6 trees where the maximum number of steps required to construct the trees is $2k+2$. Additionally, we measured the average of maximum communication steps between the root node and all other nodes in the network among all trees with the following cases: (1) no faulty node, (2) one faulty node, (3) two faulty nodes, (4) three faulty nodes, (5) four faulty nodes, and (6) five faulty nodes. We did not measure beyond 5 faulty nodes since in the worst case the root node will be pruned from the trees if all of its neighbors are faulty, and there will be no path to other nodes that can be used to measure the efficiency of the communications. That is, the root node will isolated from the network if all its neighboring nodes are faulty.

The network sizes selected in the simulation are when $\alpha = 1+2\rho$, $\alpha = 2+3\rho$, $\alpha = 3+4\rho$, $\alpha = 4+5\rho$, $\alpha = 5+6\rho$, $\alpha = 6+7\rho$, $\alpha = 7+8\rho$, $\alpha = 8+9\rho$, and $\alpha = 9+10\rho$. The results of the simulations are illustrated, in respective order, in Tables \ref{avgMaxStepsAllPort} and \ref{maxStepsAllPorts}. Both tables are represented in Figures \ref{avgMaxStepsAllPortChart} and \ref{maxMaxStepsAllPortChart}, respectively. Some values are omitted due to the hardware resource limitations. Table \ref{avgMaxStepsAllPort} shows the average maximum number of communication steps in all IST using all ports. Whereas, Table \ref{maxStepsAllPorts} shows the maximum of all maximums number of communication steps in all IST using all ports. It is observable that the simulation results are consistent with the discussions in Section \ref{sec:IST} where the results are bounded by the lower and upper bounds. The lower bound is $k+1$, whereas, the upper bound is $2k+2$ which is equal to the tree $depth - 1$. That is, one more step is counted when the last node is trying to communicate to its neighboring nodes.

The simulation measures the required number of communication steps to reach each destination node $D$ from the source node $S = 0$ in the network with no faulty node. In one faulty node, we run the simulation $n$ times, where $n$ is the total number of nodes in the network, and in each run we take one node down then we measure the required number of communication steps to reach the destination node. That is, In case of one faulty node, for each network size, we measured the maximum number of communication steps required to reach each node in the network from the root node with all one node fault possibilities and then we obtain the average and the maximum of the steps. The same simulation applied for the cases when all possibilities of 2, 3, 4, and 5 faulty nodes are present in each network size.
 
\begin{table}[H]
\centering
\caption{Average maximum number of steps to construct all trees using all ports.}
\label{avgMaxStepsAllPort}
\resizebox{\columnwidth}{!}{
\begin{tabular}{|c|c|c|c|c|c|c|c|c|c|}
\hline
$\alpha$  & 1+2$\rho$ & 2+3$\rho$ & 3+4$\rho$ & 4+5$\rho$ & 5+6$\rho$ & 6+7$\rho$ & 7+8$\rho$ & 8+9$\rho$ & 9+10$\rho$  \\ \hline
Lower Bound & 2   & 3            & 4           & 5           & 6            & 7           & 8        & 9            & 10          \\ \hline
No Faulty & 2        & 3            & 4           & 5           & 6             & 7          & 8        & 9            & 10          \\ \hline
1 Faulty  & 2         & 3.333     & 4.5         & 5.6        & 6.666     & 7.714   & 8.75    & 9.777    & 10.8        \\ \hline
2 Faulty  & 2         & 3.529     & 4.852     & 6.061    & 7.208     & 8.32     & 9.409  & 10.481  & 11.542      \\ \hline
3 Faulty  & 2         & 3.649     & 5.105     & 6.417    & 7.648     & 8.831   & 9.98    & 11.107  &             \\ \hline
4 Faulty  & 2         & 3.765     & 5.314     & 6.71      & 8.017     & 9.266   &            &              &             \\ \hline
5 Faulty  & 2         & 3.899     & 5.512     & 6.971    & 8.339     &             &            &              &             \\ \hline
Upper Bound & 4  & 6            & 8            & 10         & 12            & 14        & 16       & 18      & 20          \\ \hline
\end{tabular}
}
\end{table}

\begin{table}[H]
\centering
\caption{Maximum of all maximums number of steps to construct all trees using all ports.}
\label{maxStepsAllPorts}
\resizebox{\columnwidth}{!}{
\begin{tabular}{|c|c|c|c|c|c|c|c|c|c|}
\hline
$\alpha$  & 1+2$\rho$ & 2+3$\rho$ & 3+4$\rho$ & 4+5$\rho$ & 5+6$\rho$ & 6+7$\rho$ & 7+8$\rho$ & 8+9$\rho$ & 9+10$\rho$ \\ \hline
Lower Bound & 2   & 3         & 4         & 5         & 6       & 7        & 8         & 9        & 10          \\ \hline
No Faulty & 2        & 3         & 4         & 5         & 6       & 7         & 8        & 9        & 10         \\ \hline
1 Faulty  & 2         & 4         & 6         & 8         & 10      & 12      & 14      & 16      & 18         \\ \hline
2 Faulty  & 2         & 4         & 6         & 8         & 10      & 12      & 14      & 16      & 18         \\ \hline
3 Faulty  & 2         & 4         & 6         & 8         & 10      & 12      & 14      & 16      &            \\ \hline
4 Faulty  & 2         & 6         & 8         & 10       & 12      & 14      &           &           &            \\ \hline
5 Faulty  & 2         & 6         & 8         & 10       & 12      &           &           &           &            \\ \hline
Upper Bound & 4  & 6         & 8         & 10       & 12      & 14      & 16       & 18      & 20          \\ \hline
\end{tabular}
}
\end{table}

\begin{figure}[h]
\centering
\includegraphics[scale=0.50]{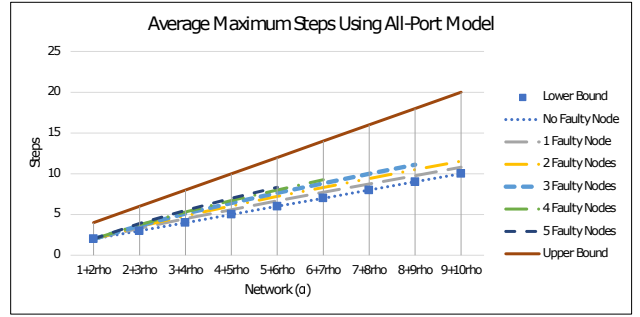}
\caption{Average maximum steps using all-port model.}
\label{avgMaxStepsAllPortChart}
\end{figure}

\begin{figure}[h]
\centering
\includegraphics[scale=0.50]{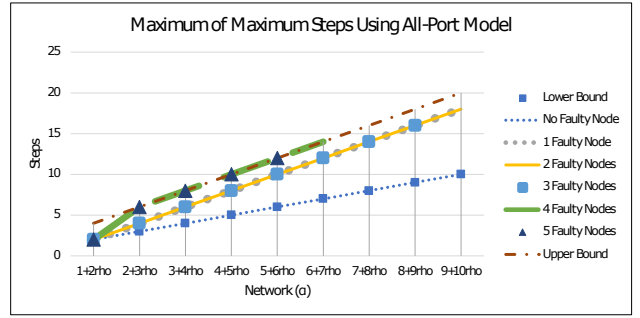}
\caption{Maximum of maximum steps using all-port model.}
\label{maxMaxStepsAllPortChart}
\end{figure}

\section{Spanning Trees in Higher Dimensional EJ Networks\label{sec:STinHigherEJ}}
In this section, we apply the proposed work on higher dimensional EJ networks \cite{hussain2017higher} to obtain the spanning trees.
The higher dimensional EJ networks are explained in Subsection \ref{sec:HigherEJ}. In Subsection \ref{sec:STHEJ}, we study the spanning trees  in higher dimensional
EJ networks.

\subsection{Higher Dimensional EJ Networks\label{sec:HigherEJ}}
The higher dimensional EJ network \cite{hussain2017higher} is denoted as $EJ_{\alpha}^{(n)}$ and it is based on the cross product between the lower dimensional EJ networks.
That is, $EJ_{\alpha}^{(n)} = EJ_{\alpha} \otimes EJ_{\alpha}^{(n-1)}$, which is $EJ_{\alpha}$ cross product itself $n$ times, where $n$ is known as the number of dimensions.
In this paper, we strict $\alpha$ to be dense, i.e., $\alpha = a + b\rho \in \mathbb{Z}[\rho]$ where $b = a+1$, and the $\alpha$ of all dimensions are not necessarily equal, i.e., same network sizes.

The result of the cross product between any two graphs $G_1(V_1,E_1)$ and $G_2(V_2,E_2)$ is $G(V,E)$. Then, $G(V,E)$ can be written as
$G_1 \times G_2$ where $V = \{(u, v) | u \in V_1, v \in V_2\}$ and
$E = \{((u_1, v_1), (u_2, v_2)) | ((u_1, u_2) \in E_1 \ and \ v_1 = v_2) \ or \ ((v_1, v_2) \in E_2 \ and \ u_1 = u_2)\}$.

The norm of $EJ_{\alpha}^{(n)}$ is $N(\alpha)^n$, which is the total number of nodes in $EJ_{\alpha}$ network power of $n$.
To address the nodes in $EJ_{\alpha}^{(n)}$, a set of $n$-tuples with coordinates in EJ is used, from the highest to the lowest dimensions.
That is, a node $(x_n + y_n\rho, x_{n-1} + y_{n-1}\rho, \dots, x_1 + y_1\rho)$ is located in the positions $x_n + y_n\rho$ on the first layer (highest or nth-dimension)
of $EJ_{\alpha}^{(n)}$, $x_{n-1} + y_{n-1}\rho$ on the second layer of $EJ_{\alpha}^{(n)}$, and so on until $x_1 + y_1\rho$ on the last layer (lowest or 1st-dimension) of $EJ_{\alpha}^{(n)}$.
In $EJ_{\alpha}^{(n)}$, each node has degree of $6n$. the network $EJ_{\alpha}^{(n)}$ can be represented by placing a copy of $EJ_{\alpha}^{(n-1)}$ on each node of EJ.
For example, Fig. \ref{23rho_2D} shows the network $EJ^{(2)}_{2+3\rho}$ and the edges of the black node $(1-\rho^2, 1+\rho)$ are connected to its neighbores, and the neighbors of node $(0,0)$ are obvious.

\begin{figure}[h]
\centering
\includegraphics[scale=0.50]{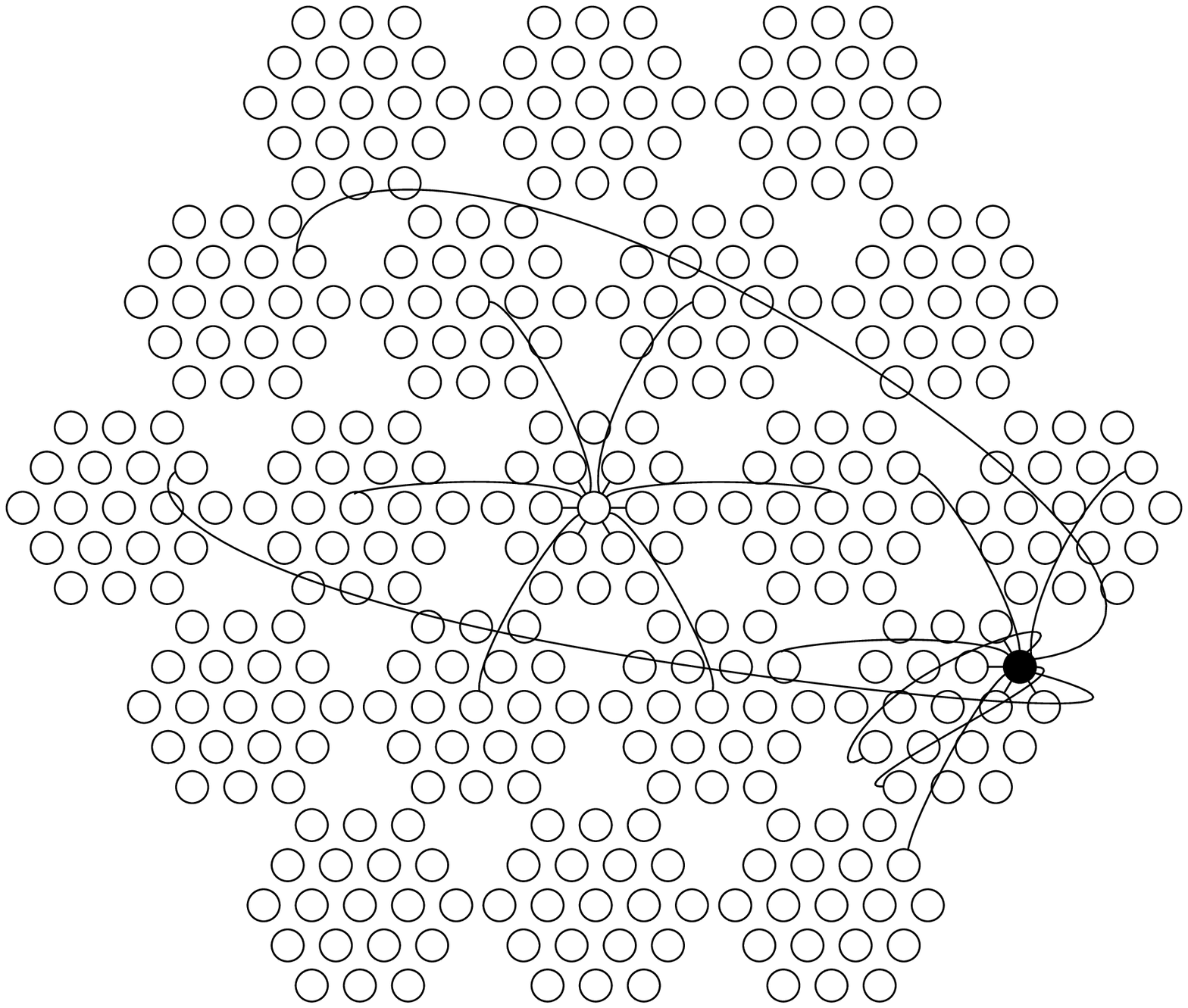}
\caption{$EJ_{2+3\rho}^{(2)}$.}
\label{23rho_2D}
\end{figure}

\subsection{Spanning Trees in $EJ_{\alpha}^{(n)}$\label{sec:STHEJ}}
In this subsection, we explain the construction of the spanning trees in $EJ_{\alpha}^{(n)}$.



In order to obtain the 3 edge disjoint node independent spanning trees, we can recursively apply the tree construction method discussed in Section \ref{sec:EDNIST} on the higher dimensional EJ networks.
That is, the proposed construction method is applied on each dimension (layer) of $EJ_{\alpha}^{(n)}$ (from the highest layer to the lowest layer). For instance, the $EJ_{2+3\rho}^{(2)}$ is composed of two layers. The tree construction method is performed on the first layer, and whenever the node in the first layer has a link then it can recursively apply the tree construction method on the second layer of the network. The same approach can be followed to obtain 6 node independent spanning trees in $EJ_{\alpha}^{(n)}$ by recursively applying the tree construction method discussed in Section \ref{sec:IST}.
The below algorithm describes the tree construction.

\begin{algorithm}[H]
\caption{ConstructSTonHigherEJ($EJ_{\alpha}^{(n)}$)}
\begin{algorithmic}[1]
\label{ConstructSTonHigherEJ}
\FOR{$i = n$ to $1$}
    \STATE Apply construction algorithm in Section \ref{sec:EDNIST} (or \ref{sec:IST}) on $i^{th}$ layer
\ENDFOR
\end{algorithmic}
\end{algorithm}

Figures \ref{2DEDNIST} and \ref{2DIST} illustrate the first edge disjoint node independent spanning trees and the first node independent spanning trees in $EJ_{2+3\rho}^{(2)}$, respectively.
The other spanning trees can be obtained by applying the rotations based on the their corresponding sections.

\begin{figure}[h]
\centering
\includegraphics[scale=0.40]{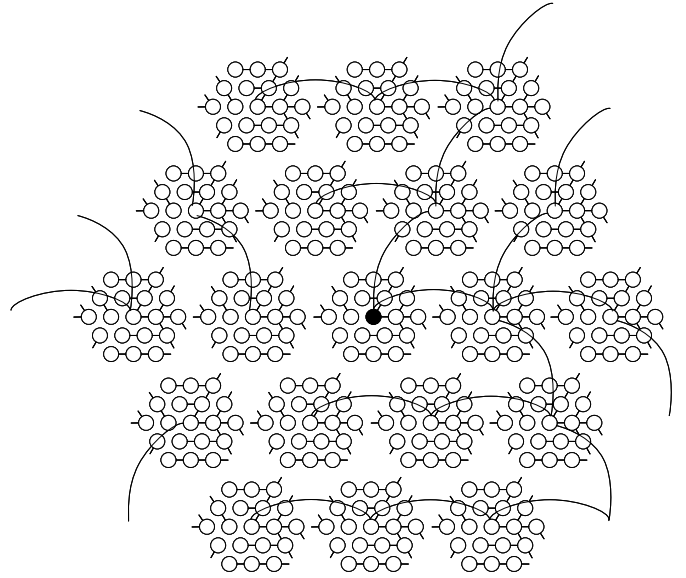}
\caption{First edge disjoint node independent spanning trees in $EJ_{2+3\rho}^{(2)}$.}
\label{2DEDNIST}
\end{figure}

\begin{figure}[h]
\centering
\includegraphics[scale=0.40]{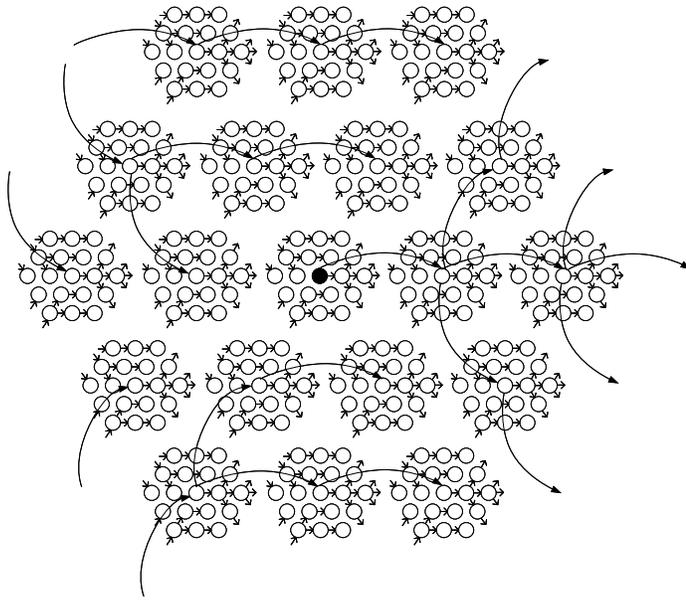}
\caption{First node independent spanning trees in $EJ_{2+3\rho}^{(2)}$.}
\label{2DIST}
\end{figure}

\section{Conclusion\label{sec:conclusion}}
In this paper, we have presented two construction techniques of edge-disjoint node-independent spanning trees (EDNIST) and node-independent spanning trees (IST) in Eisenstein-Jacobi networks.
Because of the network symmetry, in EDNIST, the first tree is constructed and then it is rotated twice to obtain the second and third disjoint trees. Whereas in IST, the first tree is constructed and then it
is rotated five times to get the second, third, forth, fifth, and sixth independent spanning trees. We have shown that the depth of EDNIST is $2k+2$ and the depth of IST is $2k+1$. Additionally, for borth trees, we
have presented a unified routing algorithm for a given network, source node $S$, and a destination node $D$. The complexity of the routing algorithm in the source node is $O(n+k)$ and none for the intermediate nodes.
The communication complexity is equal to the depth of the corresponding tree.

The simulation presented in Section \ref{sec:experimentalResults} supports the Lemmas and Theorems proved in this paper. The simulation shows the average maximum number of steps taken to construct all trees using all ports simultaneously
with no faulty, 1 faulty, 2 faulty, 3 faulty, 4 faulty, and 5 faulty nodes. Further, the maximum of all maximums number of steps to construct all trees using all ports simultaneously is bounded to the upper bound.

For future work, we will further investigate the problem to find parallel constructions for
both EDNIST and IST. Furthermore, we will also investigate whether there are more than 3 and 6, in respective order, edge disjoint node independent spanning trees and node independent spanning trees in higher dimensional Eisenstein-Jacobi networks.

\bibliographystyle{IEEEtranS}
\bibliography{references}

\end{document}